\documentclass[english]{IEEEtran}
\usepackage[T1]{fontenc}
\usepackage[latin9]{inputenc}
\usepackage{amsthm}
\usepackage{amsmath}
\usepackage{graphicx}
\usepackage{amssymb}
\usepackage{esint}

\makeatletter
  \theoremstyle{definition}
  \newtheorem{defn}{Definition}
\theoremstyle{plain}
\newtheorem{thm}{Theorem}
  \theoremstyle{remark}
  \newtheorem{rem}{Remark}
  \theoremstyle{plain}
  \newtheorem{cor}{Corollary}
  \theoremstyle{plain}
  \newtheorem{lem}{Lemma}
 \theoremstyle{definition}
  \newtheorem{example}{Example}

\usepackage{cite}
\usepackage{times}

\makeatother

\usepackage{babel}

\begin{document}

\title{Multi-tier Network Performance Analysis using a Shotgun Cellular
System}

\author{{\normalsize Prasanna Madhusudhanan$^{*}$, Juan G. Restrepo$^{\dagger}$,
Youjian (Eugene) Liu$^{*}$, Timothy X Brown}$^{*+}${\normalsize ,
and Kenneth R. Baker}$^{+}${\normalsize \\$^{*}$ Department of
Electrical, Computer and Energy Engineering, $^{\dagger}$ Department
of Applied Mathematics, $^{+}$ Interdisciplinary Telecommunications
Program\\University of Colorado, Boulder, CO 80309-0425 USA\\\{mprasanna,
juanga, eugeneliu, timxb, ken.baker\}@colorado.edu}}
\maketitle
\begin{abstract}
This paper studies the carrier-to-interference ratio $\left(\frac{C}{I}\right)$
and carrier-to-interference-plus-noise ratio $\left(\frac{C}{I+N}\right)$
performance at the mobile station (MS) within a multi-tier network
composed of $M$ tiers of wireless networks, with each tier modeled
as the \textit{homogeneous} n-dimensional (n-D, n=1,2, and 3) shotgun
cellular system, where the base station (BS) distribution is given
by the homogeneous Poisson point process in n-D. The $\frac{C}{I}$
and $\frac{C}{I+N}$ at the MS in a single tier network are thoroughly
analyzed to simplify the analysis of the multi-tier network. For the
multi-tier network with given system parameters, the following are
the main results of this paper: (1) semi-analytical expressions for
the tail probabilities of $\frac{C}{I}$ and $\frac{C}{I+N}$; (2)
a closed form expression for the tail probability of $\frac{C}{I}$
in the range {[}1,$\infty$); (3) a closed form expression for the
tail probability of an approximation to $\frac{C}{I}$ in the entire
range {[}0,$\infty$); (4) a lookup table based approach for obtaining
the tail probability of $\frac{C}{I+N}$, and (5) the study of the
effect of shadow fading and BSs with ideal sectorized antennas on
the $\frac{C}{I}$ and $\frac{C}{I+N}$. Based on these results, it
is shown that, in a practical cellular system, the installation of
additional wireless networks (microcells, picocells and femtocells)
with low power BSs over the already existing macrocell network will
always improve the $\frac{C}{I+N}$ performance at the MS.\end{abstract}

\section{Introduction\label{sec:Introduction}}

The modern cellular communication network is a complex overlay of
heterogeneous networks such as macrocells, microcells, picocells,
femtocells, etc. The base station (BS) distribution appears increasingly
irregular as the density of BSs grows over time while bounded by cell
site limitation. Due to computational constraints, system designers
cannot study the overall network at once, and have to resort to simulations
for specific portions of the network. As it is hard to obtain insight
and general conclusions from such studies, it is desirable to abstract
and simplify the model. At one end of the abstraction, the BSs are
assumed to be at the centers of regular hexagonal cells. At the other
end, the BS deployments are modeled according to a Poisson point process.
In \cite{Brown2000}, the author makes a connection between the ideal
hexagonal cellular system and the cellular system with the BS placement
according to a homogeneous Poisson point process on a plane (two dimensions,
2-D), called the shotgun cellular system (SCS). It is shown that the
carrier-to-interference ratio, $\left(\frac{C}{I}\right)$, of the
SCS lower bounds that of the ideal hexagonal cellular system and moreover,
they converge in the strong shadow fading regime. We have explored
the SCS in detail in \cite{Madh0000:Carrier,Madhusudhanan2010a,Brown2000}.
The utility of the SCS model in the study of the cognitive radio networks
can be found in \cite{Madhusudhanan2010}.

In this paper, we study the practical cellular system by viewing the
macrocells, microcells, picocells and femtocells as the different
tiers of a multi-tier network. We focus on the $\frac{C}{I}$ and
the carrier-to-interference-plus-noise ratio $\left(\frac{C}{I+N}\right)$
at the mobile station (MS) in a multi-tier network with $M$ tiers
of heterogeneous networks (hence called an $M$-tier network). The
BS distribution of the practical cellular system follows regular
topologies (e.g. to match the customer density patterns along highways,
between suburbs and city centers and within large multi-storey buildings).
Each tier of the $M$-tier network is modeled as the \textit{homogeneous}
$l$- dimensional $\left(l-\mathrm{D,}\ l=1,2,\mathrm{\ and\ }3\right)$
SCS, where the BS distribution is according to the homogeneous Poisson
point process in $\mathbb{R}^{l},$ $l=1,2,3$. In the \textit{homogeneous}
$l$-D SCS, $l$=1 is a model for the highway scenario, $l$=2 models
the planar deployment of BSs in suburbs, and $l$=3 models the BS
deployments within large multi-storey buildings and wireless LANs
(WLAN) in muti-storey residential areas. A Poisson point process in
$\mathbb{R}^{2}$ has been a popular model adopted in the literature
for the locations of nodes in the study of ad hoc and other uncoordinated
networks (\cite{WebYan2005,Zorzi1994,Takagi1984} are a few selected
references). It has also been used in studying two-tier networks composed
of macrocells and femtocells \cite{Chandrasekhar2009b,Xia2010}. Here,
we characterize the cellular performance in a multi-tier network with
BS distributions according to the Poisson point process in $\mathbb{R}^{1}$,
$\mathbb{R}^{2}$ and $\mathbb{R}^{3}$. In \cite{Dhillon2011}, the
authors study the multi-tier network with the BS distribution in the
various tiers according to the Poisson point process in $\mathbb{R}^{2},$
and derive a closed form expression for the tail probability of $\frac{C}{I}$
in the range $[1,\infty)$ for the special case of Rayleigh fading.
In this paper, we characterize the $\frac{C}{I}$ in the entire range
$[0,\infty)$ and for any general fading distribution.

\subsubsection*{Contributions of the paper}

Firstly, we emphasize that the study of the cellular performance of
the multi-tier network is tightly coupled with a similar study on
a single tier network. Hence, we indulge in thoroughly understanding
the single tier network and its properties. Sections \ref{sec:ctoi_singleTier},
\ref{sec:noise} and \ref{sec:sf} deal with the single tier network.
In Section \ref{sec:multiTier}, based on the theory developed in
the previous sections, we completely characterize the signal quality
at the MS in a $M$-tier network measured in terms of the carrier-to-interference
ratio $\left(\frac{C}{I}\right)$ and the carrier-to-interference-plus-noise
ratio $\left(\frac{C}{I+N}\right)$. In particular, for the multi-tier
network, we derive (1) semi-analytical expressions for the tail probabilities
of $\frac{C}{I}$ and $\frac{C}{I+N}$; (2) a closed form expression
for the tail probability of $\frac{C}{I}$ in the range $\left[1,\infty\right)$;
(3) a closed form expression for the tail probability of an approximation
to $\frac{C}{I}$ in the entire range {[}0,$\infty$); (4) a lookup
table based approach for obtaining the tail probability of $\frac{C}{I+N}$,
and (5) the effect of shadow fading and BSs with ideal sectorized
antennas on the $\frac{C}{I}$ and $\frac{C}{I+N}$. Finally, it is
shown that the installation of additional wireless networks (microcells,
picocells and femtocells) with low power BSs over the already existing
macrocell network will always improve the $\frac{C}{I+N}$ performance
at the MS.

\section{System Model\label{sec:modelreview}}

This section describes the various elements used to model the wireless
network, namely, the BS layout, the radio environment, and the performance
metrics of interest.

\subsubsection*{BS Layout}

We define the SCS and \textit{homogeneous} $l$-D SCS ($l$=1,2,3)
and describe the model for the single tier and multi-tier networks.
\begin{defn}
The \textit{Shotgun Cellular System (SCS) }is a model for the cellular
system in which the BSs are placed in a given $l$-dimensional plane
($l=1,2,\ \mathrm{and}\ 3$) according to a Poisson point process
on $\mathbb{R}^{l}$. The intensity function of the Poisson point
process is called the BS density function in the context of the SCS.
(See \cite{Madhusudhanan2010a} for more details.)
\end{defn}

\begin{defn}
In the \textit{homogeneous} $l$-D SCS $\left(l\in\left\{ 1,2,3\right\} \right)$,
the BSs are placed according to a homogeneous Poisson point process
on $\mathbb{R}^{l}$ with a BS density $\lambda_{0},$ such that the
probability that there exists a BS in a small region $\mathcal{H}\subseteq\mathbb{R}^{l}$
is $\lambda_{0}\left|\mathcal{H}\right|,$ where $\left|\mathcal{H}\right|\ll1$
is the length, area or volume of the region $\mathcal{H}$ for $l=1,2,\mathrm{and}\ 3$,
respectively; and the events in non-overlapping regions are independent
of each other.
\end{defn}

\subsubsection*{Radio Environment}

The signal from the BS undergoes path-loss and shadow fading; and
is also affected by background noise. The signal power at a distance
$R$ from the BS is given by $P=K\Psi R^{-\varepsilon},$ where $K$
captures the transmission power and the antenna gain of the BS, $\Psi$
is the random shadow fading factor, and $R^{-\varepsilon}$ represents
the inverse power law path-loss with $\varepsilon$ as the path-loss
exponent, and $R$ as the distance from the BS. The noise power in
the system is $N$.

\subsubsection*{Single tier network}

In this paper, the single-tier network refers to the macrocell network
and the BS layout is according to the \textit{homogeneous} $l$-D
SCS, $l=1,2,3$.

\subsubsection*{Multi-tier (M-tier) network}

The $M$-tier network is assumed to be composed of $M$ independent
\textit{homogeneous} $l$-D SCSs with BS density $\left\{ \lambda_{i}\right\} _{i=1}^{M},$
for each tier.  For the $M$-tier network, $K$ and the cumulative
density function (c.d.f.) of $\Psi$ are different for each tier.

\subsubsection*{Performance Metric}

In this paper, we are concerned with the signal quality at a MS within
the wireless network. The MS is assumed to be located at the origin
of $\mathbb{R}^{l},\ l=1,2,3$ in which the multi-tier network is
defined. The MS receives signals from all the BSs, and chooses to
communicate with the BS that corresponds to the strongest received
signal power. This BS is referred to as the {}``serving BS'', and
all the other BSs are collectively called the {}``interfering BSs''.
Consequently, the signal quality at the MS is defined as the ratio
of the received power from the serving BS $(\mbox{denoted by }C\ \mathrm{or}\ P_{S})$
to the sum of the total interference power ($\mbox{denoted by }$
$I\ \mathrm{or}\ P_{I},$ sum of the powers from the interfering BSs)
and the noise power $(N)$, and is called the carrier-to-interference-plus-noise
ratio $\left(\frac{C}{I+N}\right)$. In an {}``interference limited
system'', $I\gg N$ and the signal quality is referred to as the
carrier-to-interference ratio $\left(\frac{C}{I}\right)$. Thus, for
a single tier network, the $\frac{C}{I}$ and $\frac{C}{I+N}$ are
\begin{equation}
\frac{C}{I}\overset{\left(a\right)}{=}\frac{K_{S}\Psi_{S}R_{S}^{-\varepsilon}}{\sum_{i=1}^{\infty}K_{i}\Psi_{i}R_{i}^{-\varepsilon}},\ \frac{C}{I+N}\overset{\left(b\right)}{=}\frac{K_{S}\Psi_{S}R_{S}^{-\varepsilon}}{\sum_{i=1}^{\infty}K_{i}\Psi_{i}R_{i}^{-\varepsilon}+N},\label{eq:ctoisf}\end{equation}

\noindent where subscript {}``$S$'' denotes the serving BS and
subscript {}``$i$'' indexes the interfering BSs; $K_{S}=\left\{ K_{i}\right\} _{i=1}^{\infty}$
are the transmission powers that can be equal a constant or independent
and identically distributed (i.i.d.) random variables; $R_{S}\ \mathrm{and}\ \{R_{i}\}_{i=1}^{\infty}$
are random variables that come from the underlying Poisson point process
that governs the BS placement; $\Psi_{S}\ \mathrm{and}\ \{\Psi_{i}\}_{i=1}^{\infty}$
are i.i.d. random variables. Hence, $\frac{C}{I}$ and $\frac{C}{I+N}$
are random variables, and can be characterized by a probability density
function (p.d.f.), c.d.f. or the tail probability. The tail probability
of $\frac{C}{I+N}$ is given by $\mathbb{P}\left(\left\{ \frac{C}{I+N}>\eta\right\} \right)$,
and is the probability that a MS in the SCS has a signal quality of
at least $\eta,\ \eta\ge0$. In the following section, we characterize
the tail probability of the $\frac{C}{I}$ at the MS in a single tier
network.

\section{$\frac{C}{I}$ characterization for a single tier network\label{sec:ctoi_singleTier}}

Here, the transmission power and the antenna gains of all the BSs
in the SCS are assumed to be constant $(\mathrm{say,\ }K)$. Also,
the shadow fading factors are assumed to be unity. Hence, from the
expression for $\frac{C}{I}$ in $\left(\ref{eq:ctoisf}a\right)$,
the BS closest to the MS is the serving BS and the expression for
$\frac{C}{I}$ is \begin{equation}
\frac{C}{I}=\frac{KR_{1}^{-\varepsilon}}{\sum_{i=2}^{\infty}KR_{i}^{-\varepsilon}},\label{eq:ctoinosf}\end{equation}
where $R_{1}\le R_{2}\le R_{3}\cdots$ are the distances between the
BSs and the MS, arranged in a non-decreasing order. Further, recall
that the BS layout in the single tier network is as in the \textit{homogeneous
}$l$-D SCS $\left(l=1,2,3\right)$ with BS density $\lambda_{0}$.
Thus, the p.d.f. of $R_{1}$ is given by $f_{R_{1}}\left(r_{1}\right)=\lambda_{0}b_{l}r_{1}^{l-1}\mathrm{e}^{-\frac{\lambda_{0}b_{l}r_{1}^{l}}{l}},$
where $r_{1}\geq0$ and $b_{l}=2,2\pi,4\pi$ for $l=1,2,3$, respectively,
and the conditional p.d.f. of the $i\mathrm{^{th}}$ closest BS conditioned
on the $(i-1){}^{\mathrm{th}}$ closest BS, is $f_{R_{i}|R_{i-1}}\left(r_{i}|r_{i-1}\right)=\lambda_{0}b_{l}r_{i}^{l-1}\mathrm{e}^{-\frac{\lambda_{0}b_{l}\left(r_{i}^{l}-r_{i-1}^{l}\right)}{l}},$
$r_{i}\geq r_{i-1}$.
\begin{thm}
\label{thm:ldscs} \emph{In a homogeneous $l$-D SCS with a constant
BS density $\lambda_{0}$, if the path-loss exponent satisfies $\varepsilon>l$, }

\emph{(a) the characteristic function of $P_{I}$ conditioned on $R_{1}$
is \begin{eqnarray}
 &  & \Phi_{P_{I}|R_{1}}\left(\omega|r_{1}\right)\nonumber \\
 &  & =\exp\left(\frac{\lambda_{0}b_{l}r_{1}^{l}}{l}\left(1-_{1}F_{1}\left(-\frac{l}{\varepsilon};1-\frac{l}{\varepsilon};\frac{i\omega K}{r_{1}^{\varepsilon}}\right)\right)\right),\label{eq:charfnPi2}\end{eqnarray}
}

\emph{(b) the characteristic function of $\left(\frac{C}{I}\right)^{-1}$
is given by \begin{eqnarray}
\Phi_{\left(\frac{C}{I}\right)^{-1}}\left(\omega\right) & = & \mathrm{E}_{R_{1}}\left[\Phi_{P_{I}|R_{1}}\left(\left.\frac{\omega}{P_{S}}\right|R_{1}\right)\right]\nonumber \\
 & = & \frac{1}{_{1}F_{1}\left(-\frac{l}{\varepsilon};1-\frac{l}{\varepsilon};i\omega\right)},\label{eq:charfnitoc2}\end{eqnarray}
where }\textup{\emph{$\mathrm{E}_{R_{1}}$}}\emph{ is the expectation
w.r.t. $R_{1}$, and $_{1}F_{1}\left(\cdot;\cdot;\cdot\right)$ is
called the confluent hypergeometric function of the first kind .}\end{thm}
\begin{proof}
See \cite[Corollary 2]{Madhusudhanan2010a}.
\end{proof}
The significance of Theorem \ref{thm:ldscs} is in the following remarks.
\begin{rem}
\label{rem:ctoiTailProb}The tail probability of $\frac{C}{I}$ may
be directly obtained from the characteristic function and is given
by\begin{eqnarray}
 &  & \mathbb{P}\left(\left\{ \frac{C}{I}>\eta\right\} \right)\nonumber \\
 &  & =\begin{cases}
\int_{\omega=-\infty}^{\infty}\Phi_{\left(\frac{C}{I}\right)^{-1}}\left(\omega\right)\left(\frac{1-\exp\left(-\frac{i\omega}{\eta}\right)}{i\omega}\right)\frac{d\omega}{2\pi}, & \eta>0\\
1, & \eta=0.\end{cases}\label{eq:nosf-expr}\end{eqnarray}
\end{rem}
\begin{proof}
See \cite[Eq. (9)]{Madh0000:Carrier}.
\end{proof}

\begin{rem}
\label{rem:ctoiIndependentOfLambda}The characteristic function of
the $\left(\frac{C}{I}\right)^{-1}$ does not depend on $\lambda_{0}$,
and hence the tail probability of $\frac{C}{I}$ at a MS in the \textit{homogeneous}
$l$-D SCS does not depend on $\lambda_{0}$.
\end{rem}

\begin{rem}
\label{rem:ctoiAndEpsilon}The characteristic function of $\left(\frac{C}{I}\right)^{-1}$
for a \textit{homogeneous} 2-D and 3-D SCS is the same as that of
a \textit{homogeneous} 1-D SCS with path-loss exponents $\frac{\varepsilon}{2}$
and $\frac{\varepsilon}{3}$, respectively. Hence, the corresponding
$\frac{C}{I}$ performances are identical.

Remark \ref{rem:ctoiIndependentOfLambda} proves why the curves corresponding
to the \textit{homogeneous} 1-D, 2-D and 3-D SCSs in Fig. \ref{fig:uniform_scs_comparison}
are straight lines. Remark \ref{rem:ctoiAndEpsilon} helps build an
intuition of why the \textit{homogeneous} 1-D SCS has a higher tail
probability of $\frac{C}{I}$ than \textit{homogeneous} 2-D and 3-D
SCSs; Fig. \ref{fig:uniform_scs_comparison} now corresponds to comparing
the tail probabilities of $\frac{C}{I}$ in a \textit{homogeneous}
1-D SCS with path-loss exponents $\varepsilon,\ \frac{\varepsilon}{2},\ \mathrm{and}\ \frac{\varepsilon}{3}$,
respectively. As the path-loss exponent decreases, the BSs farther
away from the MS have a greater contribution to the total interference
power at the MS, and this leads to a poorer $\frac{C}{I}$ at the
MS and a smaller tail probability (computed by evaluating the integral
in $\left(\ref{eq:nosf-expr}\right)$). An important consequence of
Remark \ref{rem:ctoiAndEpsilon} is as follows.\end{rem}
\begin{cor}
\label{cor:ctoiClosedForm}\emph{For a homogeneous $l$-D SCS, $l=1,2,3$,
where the path-loss exponent is $\varepsilon$, the tail probability
of $\frac{C}{I}$ is \begin{eqnarray}
\mathbb{P}\left(\left\{ \frac{C}{I}>\eta\right\} \right) & = & \mathbb{P}\left(\left\{ \frac{C}{I}>1\right\} \right)\times\eta^{-\frac{l}{\varepsilon}}\label{eq:ctoiTailGt1-3}\\
 & = & \mathcal{K}_{\frac{\varepsilon}{l}}\eta^{-\frac{l}{\varepsilon}},\ \forall\ \eta\ge1,\ \varepsilon>l,\label{eq:ctoiTailGt1-2}\end{eqnarray}
where $\mathcal{K}_{\frac{\varepsilon}{l}}$ is a constant parametrized
by $\frac{\varepsilon}{l}$.}\end{cor}
\begin{proof}
In \cite{Brown2000}, we have shown that \begin{eqnarray}
\mathbb{P}\left(\left\{ \frac{C}{I}>\eta\right\} \right) & = & \mathbb{P}\left(\left\{ \frac{C}{I}>1\right\} \right)\times\eta^{-\frac{2}{\varepsilon}},\label{eq:ctoiTailGt1}\end{eqnarray}
where $\eta\ge1,$ and $\varepsilon>2$, for a \textit{homogeneous}
2-D SCS. From Remark \ref{rem:ctoiAndEpsilon}, $\left(\ref{eq:ctoiTailGt1}\right)$
holds for all \textit{homogeneous} 1-D SCS with path-loss exponent
$\frac{\varepsilon}{2}$ and therefore, for all \textit{homogeneous}
3-D SCSs with path-loss exponent $\frac{3\varepsilon}{2}$. Hence,
$\left(\ref{eq:ctoiTailGt1-3}\right)$ hold true, and $\left(\ref{eq:ctoiTailGt1-2}\right)$
is obtained by noting that the characteristic function of $\left(\frac{C}{I}\right)^{-1}$
is a function of $\frac{\varepsilon}{l}$ and so $\mathbb{P}\left(\left\{ \frac{C}{I}>1\right\} \right)$
is a constant.
\end{proof}
\begin{figure}
\begin{centering}
\includegraphics[scale=0.7]{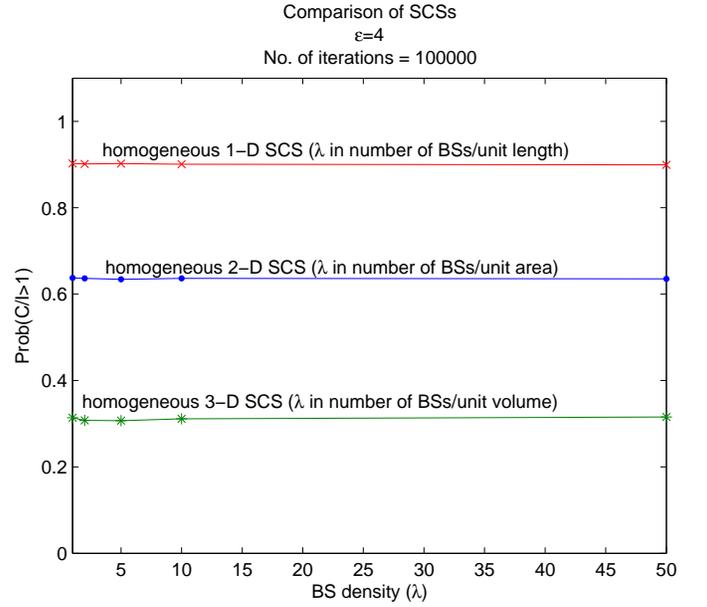}
\par\end{centering}

\caption{\label{fig:uniform_scs_comparison}Invariance of $\frac{C}{I}$ of
the \textit{homogeneous} $l$-D SCS w.r.t. BS density $\left(\lambda\right)$.
The $\frac{C}{I}$ tail probability is independent of $\lambda$ (proved
in Remark \ref{rem:ctoiIndependentOfLambda}), and the \textit{homogeneous}
1-D SCS has a better $\frac{C}{I}$ tail probability than the \textit{homogeneous}
2-D and 3-D SCS (proved in Remark \ref{rem:ctoiAndEpsilon}).}

\end{figure}
The constant $\mathcal{K}_{\frac{\varepsilon}{l}}$ can be obtained
by using $\left(\ref{eq:charfnitoc2}\right)$ and evaluating the integral
in $\left(\ref{eq:nosf-expr}\right)$ with $\eta=1$. Note that $\frac{C}{I}$
is a non-negative random variable with a support of $\left[0,\infty\right)$,
and surprisingly, its tail probability has such a simple form as given
by $\left(\ref{eq:ctoiTailGt1-2}\right)$ in the region $\left[1,\infty\right)$.
Next, we define the so-called \textit{few BS approximation} and derive
closed form expressions for the tail probability of $\frac{C}{I}$
at MS in a \textit{homogeneous} $l$-D SCS for both the regions $\left[0,1\right)$
and $\left[1,\infty\right)$.
\begin{defn}
\textit{The few BS approximation} corresponds to modeling the total
interference power at the MS in the SCS as the sum of the contributions
from the strongest few interfering BSs and an ensemble average of
the contributions of the rest of the interfering BSs.
\end{defn}
Recall from $\left(\ref{eq:ctoinosf}\right)$ that $P_{I}=\sum_{i=2}^{\infty}KR_{i}^{-\varepsilon}$,
where $\left\{ R_{i}\right\} _{i=1}^{\infty}$ is the set of distances
of BSs arranged in the ascending order of their separation from the
MS. The arrangement also corresponds to the descending order of their
contribution to $P_{I}$. In the few BS approximation, $P_{I}$ is
approximated by $\tilde{P_{I}}\left(k\right)=\sum_{i=2}^{k}KR_{i}^{-\varepsilon}+E\left[\left.\sum_{i=k+1}^{\infty}KR_{i}^{-\varepsilon}\right|R_{k}\right],$
for some $k$, where $E\left[\cdot\right]$ is the expectation operator
and refers to the ensemble average of the contributions of BSs beyond
$R_{k}$. The $\frac{C}{I}$ at the MS obtained by the few BS approximation
is denoted by $\frac{C}{I_{k}}$. Next, we study $\frac{C}{I_{k}}$
for the \textit{homogeneous} $l$-D SCSs.
\begin{lem}
\label{lem:fewBSMeanCor}For the homogeneous $l$-D SCS, with BS density
$\lambda_{0}$ and $\varepsilon>l$, for k=1,2,3, \begin{eqnarray}
E\left[\left.\sum_{i=k+1}^{\infty}KR_{i}^{-\varepsilon}\right|R_{k}\right] & = & \frac{\lambda_{0}b_{l}KR_{k}^{l-\varepsilon}}{\varepsilon-l}.\label{eq:fewBSappMean}\end{eqnarray}
\end{lem}
\begin{proof}
See \cite[Corollary 4]{Madhusudhanan2010a}.
\end{proof}
Next, the tail probability of $\frac{C}{I_{2}}=\left.\frac{C}{I_{k}}\right|_{k=2}$
is derived.
\begin{thm}
\label{thm:fewBS}\emph{In the homogeneous $l$-D SCS with BS density
$\lambda_{0}$ and path-loss exponent $\varepsilon$ $\left(\varepsilon>l\right)$,
the tail probability of $\frac{C}{I_{2}}$ is \begin{eqnarray}
 &  & \mathbb{P}\left(\left\{ \frac{C}{I_{2}}>\eta\right\} \right)\label{eq:few-bs-approx}\\
 &  & =\left\{ \begin{array}{ll}
\eta^{-\frac{l}{\varepsilon}}C_{\frac{\varepsilon}{l}}, & \eta\geq1\\
1-\frac{(1+u(\eta))}{\mathrm{e}^{u(\eta)}}+\eta^{-\frac{l}{\varepsilon}}D_{\frac{\varepsilon}{l}}\left(\eta\right), & \eta<1\end{array}\right.,\end{eqnarray}
where $u\left(\eta\right)=\left(\frac{\varepsilon}{l}-1\right)\left(\frac{1}{\eta}-1\right),$
$C_{\frac{\varepsilon}{l}}=G(0,\infty),$ $D_{\frac{\varepsilon}{l}}\left(\eta\right)=G(u(\eta),\infty),$
and $G\left(a,b\right)=\int_{v=a}^{b}\frac{v\mathrm{e}^{-v}}{\left(1+v\left(\frac{\varepsilon}{l}-1\right)^{-1}\right)^{\frac{l}{\varepsilon}}}dv$.}\end{thm}
\begin{proof}
See \cite[Theorem 2]{Madhusudhanan2010a}.
\end{proof}
Notice that $\mathbb{P}\left(\left\{ \frac{C}{I}>\eta\right\} \right)=\frac{K_{\frac{\varepsilon}{l}}}{C_{\frac{\varepsilon}{l}}}\mathbb{P}\left(\left\{ \frac{C}{I_{2}}>\eta\right\} \right)$
for $\eta\ge1$.%
\begin{figure}
\centering{}\includegraphics[clip,scale=0.7]{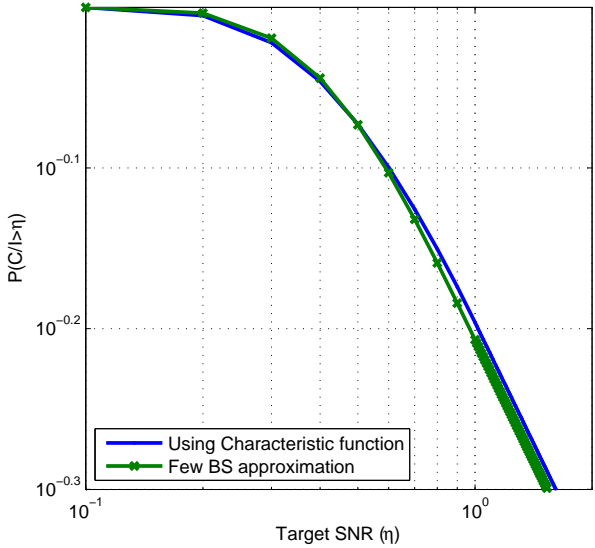}\caption{\label{fig:CompareCtoiCtoi2}\textit{Homogeneous} 2-D SCS: Comparing
$\frac{C}{I}$ and $\frac{C}{I_{2}}$ , $\varepsilon=4$}

\end{figure}
 Fig. \ref{fig:CompareCtoiCtoi2} shows the comparison of the tail
probabilities of $\frac{C}{I}$ (computed using the characteristic
function of $\left(\frac{C}{I}\right)^{-1}$) and $\frac{C}{I_{2}}$
for the \textit{homogeneous} 2-D SCS with path-loss exponent 4. Notice
that the gap between the two tail probability curves is negligible
in the region $\eta\in\left[0,1\right]$, and further, both the curves
are straight lines parallel to each other in the region $\eta\in\left[1,\infty\right)$,
when the tail probability is plotted against $\eta$, both in the
logarithmic scale. This shows that $\frac{C}{I_{2}}$ is a good approximation
for $\frac{C}{I}$ and can be characterized in closed form.

\section{$\frac{C}{I+N}$ in a single tier network\label{sec:noise}}

Here, as in Section \ref{sec:ctoi_singleTier}, the transmission powers
of all BSs are constant and shadow fading factors are equal to unity.
We first obtain the tail probability of $\frac{C}{I+N}$ using the
characteristic function of $\left(\frac{C}{I+N}\right)^{-1}$ derived
in the following corollary.
\begin{cor}
\emph{\label{cor:INtoCcharfn}In a homogeneous $l$-D SCS with BS
density $\lambda_{0}$ and path-loss exponent $\varepsilon$ $\left(\varepsilon>l\right)$,
the characteristic function of the sum of the total interference power
$\left(P_{I}\right)$ and noise power $\left(N\right)$ conditioned
on $R_{1}$ is $\Phi_{\left.P_{I}+N\right|R_{1}}\left(\left.\omega\right|r_{1}\right)=\exp\left(i\omega N\right)\times\Phi_{\left.P_{I}\right|R_{1}}\left(\left.\omega\right|r_{1}\right)$,
and the characteristic function of $\left(\frac{C}{I+N}\right)^{-1}$
is }\textup{\emph{$\Phi_{\left(\frac{C}{I+N}\right)^{-1}}\left(\omega\right)=\mathrm{E}_{R_{1}}\left[\exp\left(i\frac{\omega}{P_{S}}N\right)\times\Phi_{P_{I}|R_{1}}\left(\left.\frac{\omega}{P_{S}}\right|R_{1}\right)\right]$,
}}\emph{where $\Phi_{\left.P_{I}\right|R_{1}}\left(\left.\omega\right|r_{1}\right)$
is given by $\left(\ref{eq:charfnPi2}\right)$.}\end{cor}
\begin{proof}
The expressions for $\Phi_{\left.P_{I}+N\right|R_{1}}\left(\left.\omega\right|r_{1}\right)$
and $\Phi_{\left(\frac{C}{I+N}\right)^{-1}}\left(\omega\right)$ follow
directly from the definition of characteristic function, where $N$
is a constant.
\end{proof}
Further, the tail probability of $\frac{C}{I+N}$ is obtained by substituting
$\frac{C}{I}$ with $\frac{C}{I+N}$ in $\left(\ref{eq:nosf-expr}\right)$.
Next, an interesting property of the $\frac{C}{I+N}$ at the MS in
the \textit{homogeneous} $l$-D SCS is presented.
\begin{cor}
\label{cor:ctoiInvariance}\emph{If the $\frac{C}{I+N}$ at the MS
in the homogeneous $l$-D SCS is specified by $\left(\lambda_{0},\varepsilon,K,N\right)$
where $\lambda_{0}$ is the BS density, $\varepsilon$ is the path-loss
exponent, $K$ is the constant transmission power of each BS, and
$N$ is the constant noise power, then, \begin{eqnarray}
\left.\frac{C}{I+N}\right|_{\left(\lambda_{0},\varepsilon,K,N\right)} & =_{\mathrm{st}} & \left.\frac{C}{I+N}\right|_{\left(1,\varepsilon,1,N'\right)},\label{eq:ctoINoiseHomogeneous}\end{eqnarray}
where $N'=N\left/\left(\lambda_{0}^{\frac{\varepsilon}{l}}K\right)\right.$
and {}``$=_{\mathrm{st}}$'' means that the c.d.f's are same.}\end{cor}
\begin{proof}
See \cite[Corollary 9]{Madhusudhanan2010a}.
\end{proof}
\begin{figure}
\begin{centering}
\includegraphics[scale=0.6]{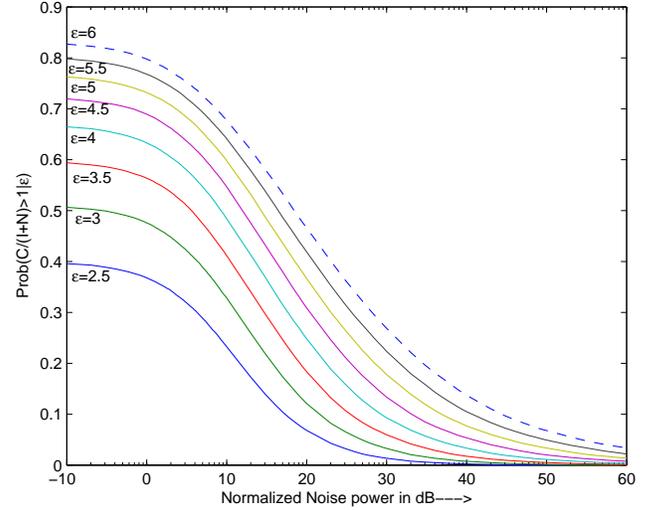}
\par\end{centering}

\caption{\label{fig:cinr_plot}$\mathbb{P}\left(\left\{ \frac{C}{I+N}>1\right\} \right)$
vs Normalized noise power $\left(N\lambda_{0}^{-\frac{\varepsilon}{l}}K^{-1}\right)$
for $l=2$}

\end{figure}
So, it is sufficient to analyze the \textit{homogeneous} $l$-D with
$\lambda_{0}=K=1$ and maintain a table for the tail probability of
$\frac{C}{I+N}$ for different values of $N'$ and $\varepsilon$.
We can find the $\frac{C}{I+N}$ at the MS for a \textit{homogeneous}
$l$-D SCS with any given $\left(\lambda_{0},\varepsilon,K,N\right)$
by just reading out the tail probability of $\frac{C}{I+N}$ corresponding
to $\varepsilon$ and $N'$ obtained using Corollary \ref{cor:ctoiInvariance}
from the lookup table. The lookup table is presented for a \textit{homogeneous}
2-D SCS in Fig. \ref{fig:cinr_plot} as a plot of $\mathbb{P}\left(\left\{ \frac{C}{I+N}>1\right\} \right)$
against $N'$ for different values of $\varepsilon$. Further, in
the \textit{homogeneous} $l$-D SCS, as $\lambda_{0}$ increases,
the noise power $N'$ of the equivalent SCS decreases according to
Corollary \ref{cor:ctoiInvariance}, and in the limit as $\lambda_{0}\rightarrow\infty$,
$N'$ approaches zero and hence $\frac{C}{I+N}\overset{D}{\rightarrow}\frac{C}{I}$,
where $\overset{D}{\rightarrow}$ corresponds to convergence in distribution.
Thus, in an {}``interference limited system'' (large $\lambda_{0}$),
the signal quality is measured in terms of $\frac{C}{I}$. Next, we
study the effect of shadow fading on the $\frac{C}{I}$ and $\frac{C}{I+N}$
at the MS in a single tier network.

\section{Shadow fading\label{sec:sf}}

Theorem \ref{thm:ShadowFadingHomogeneousInvariance} analytically
shows that the effect of the introduction of shadow fading to the
SCS is completely captured in the BS density of the \textit{homogeneous}
$l$-D SCS.
\begin{thm}
\label{thm:ShadowFadingHomogeneousInvariance}\emph{When shadow fading
in the form of i.i.d non-negative random factors, $\left\{ \Psi_{i}\right\} $,
is introduced to the homogeneous $l$-D SCS with BS density $\lambda_{0}$,
for the $\frac{C}{I}$ and $\frac{C}{I+N}$ analysis, the resulting
system is equivalent to another homogeneous $l$-D SCS with BS density
$\lambda_{0}\mathrm{E}\left[\Psi^{\frac{l}{\varepsilon}}\right]$,
as long as $\mathrm{E}\left[\Psi^{\frac{l}{\varepsilon}}\right]<\infty$.}\end{thm}
\begin{proof}
The expression for $\frac{C}{I}$ and $\frac{C}{I+N}$ in $\left(\ref{eq:ctoisf}a\right)$
and $\left(\ref{eq:ctoisf}b\right)$ may be written as $\frac{C}{I}=\frac{\bar{R}_{1}^{-\varepsilon}}{\sum_{k=2}^{\infty}\bar{R}_{k}^{-\varepsilon}}$
and $\frac{C}{I+N}=\frac{\bar{R}_{1}^{-\varepsilon}}{\sum_{k=2}^{\infty}\bar{R}_{k}^{-\varepsilon}+N}$,
where $\bar{R}_{1}=R_{S}\Psi_{S}^{-\frac{1}{\varepsilon}}$ and $\bar{R}_{k+1}\equiv R_{k}\Psi_{k}^{-\frac{1}{\varepsilon}},\ k=1,2,3\cdots$.
Now, the expression for $\frac{C}{I}$ is similar to the no shadow
fading case in $\left(\ref{eq:ctoinosf}\right)$, with the $\bar{R}$'s
replacing the $R$'s. Using the Marking theorem of Poisson process
in \cite[Page 55]{Kingman1993}, $\bar{R}=R\Psi^{-\frac{1}{\varepsilon}}$
follows the homogeneous Poisson process in $\mathbb{R}^{l}$ with
intensity $\lambda_{0}\mathrm{E}\left[\Psi^{\frac{l}{\varepsilon}}\right]$.
For a complete proof, see \cite[Theorem 4]{Madhusudhanan2010a}.
\end{proof}
Further, the $\frac{C}{I}$ and $\frac{C}{I+N}$ at the MS in the
\textit{homogeneous} $l$-D SCS with shadow fading is the same as
that in the equivalent \textit{homogeneous} $l$-D SCS where there
is no shadow fading. The following remark illustrates the consequence
of the theorem on the $\frac{C}{I}$ and $\frac{C}{I+N}$ at the MS.
\begin{rem}
\label{rem:ShadowingHomogeneousCtoi}In the \textit{homogeneous} $l$-D
SCS with BS density $\lambda_{0}$,

(a) shadow fading has no effect on the $\frac{C}{I}$ at the MS, and,

(b) the effect of shadow fading is completely captured in the noise
power term of the $\frac{C}{I+N}$.\end{rem}
\begin{proof}
Firstly, using Theorem \ref{thm:ShadowFadingHomogeneousInvariance},
it is sufficient to analyze the $\frac{C}{I}$ and $\frac{C}{I+N}$
for the \textit{homogeneous} $l$-D SCS with BS density $\lambda_{0}\mathrm{E}\left[\Psi^{\frac{l}{\varepsilon}}\right].$
Then, Remark \ref{rem:ShadowingHomogeneousCtoi}(a) follows from Remark
\ref{rem:ctoiIndependentOfLambda}. Finally, since the $\frac{C}{I+N}$
in this case has the same c.d.f. as the equivalent \textit{homogeneous}
$l$-D SCS in Corollary \ref{cor:ctoiInvariance} with $N'=NK^{-1}\left(\lambda_{0}\mathrm{E}\left[\Psi^{\frac{l}{\varepsilon}}\right]\right)^{-\frac{\varepsilon}{l}},$
Remark \ref{rem:ShadowingHomogeneousCtoi}(b) is proved.\end{proof}
\begin{example}
\label{exa:sf3dscs}Consider a \textit{homogeneous} 2-D SCS with an
average BS density $\lambda_{0}$, where each BS is affected by an
i.i.d log-normal shadow fading factor with a mean 0 and standard deviation
$\sigma$. Using Theorem \ref{thm:ShadowFadingHomogeneousInvariance},
the equivalent \textit{homogeneous} 2-D SCS has a BS density $\bar{\lambda}_{0}=\lambda_{0}\exp\left(\frac{2\sigma^{2}}{\varepsilon^{2}}\right)$.
Note that $\bar{\lambda}_{0}\ge\lambda_{0},\ \forall\ \sigma,\ \varepsilon$,
and from Remark \ref{rem:ShadowingHomogeneousCtoi}, the introduction
of shadow fading improves the $\frac{C}{I+N}$ performance at the
MS measured in terms of the tail probability of $\frac{C}{I}$.
\end{example}
Next, we study the $\frac{C}{I}$ and the $\frac{C}{I+N}$ at the
MS in a multi-tier network based on the analysis for the single tier
network modelled as the \textit{homogeneous} $l$-D SCS.

\section{Multi-tier networks (M-tier networks)\label{sec:multiTier}}

All the BSs of the $i^{\mathrm{th}}$ tier of a $M$-tier network
are assumed to have constant transmission power, $\left\{ \kappa_{i}\right\} _{i=1}^{M}$.
Firstly, the $M$-tier network is reduced to an equivalent single
tier network.
\begin{thm}
\label{thm:multiTierTheorem}\emph{Consider a multi-tier network consisting
of $M$ independent homogeneous $l$-D SCS with BS density $\left\{ \lambda_{i}\right\} _{i=1}^{M},$
such that all the BSs in $i^{\mathrm{th}}$ tier have a constant transmission
power $\kappa_{i}$, then, this multi-tier network is equivalent to
a single tier network (homogeneous $l$-D SCS) with}

\emph{(a) BS density $\lambda_{0}=\sum_{i=1}^{N}\lambda_{i}$, and,}

\emph{(b) the transmission power of each BS is an i.i.d. random variable
$K=\kappa_{i}$ with a probability $p_{i}=\frac{\lambda_{i}}{\lambda_{0}},$
$i=1,2,\cdots,M$.}\end{thm}
\begin{proof}
Theorem \ref{thm:multiTierTheorem} (a) follows directly from the
Superposition theorem of Poisson process in \cite[Page 16]{Kingman1993},
and holds true even as $M\rightarrow\infty.$ Next, consider a region
$\mathcal{H}\subseteq\mathbb{R}^{l},$ and let $\mathbb{P}_{i}\left(\mathcal{H}\right)$
denotes the probability of finding one BS belonging to the $i^{\mathrm{th}}$
tier conditioned on the event that there exists one BS in $\mathcal{H}$.
Then, $\mathbb{P}_{i}\left(\mathcal{H}\right)=\frac{\mathbb{P}\left(\left\{ N_{i}=1\right\} \right)}{\mathbb{P}\left(\left\{ N_{0}=1\right\} \right)}\times\prod_{j=1,\ j\ne i}^{M}\mathbb{P}\left(\left\{ N_{j}=0\right\} \right),$
where $\left\{ N_{i}\right\} _{i=0}^{M}$ is the set of random variables
denoting the number of BSs in $\mathcal{H}$ for the \textit{homogeneous}
$l$-D SCS with BS density $\left\{ \lambda_{i}\right\} _{i=0}^{M}$
($\lambda_{0}$ is defined in Theorem \ref{thm:multiTierTheorem}(a)),
respectively. Note that $N_{i}\sim\mathrm{Poisson}\left(\lambda_{i}\left|\mathcal{H}\right|\right)$,
where $\left|\mathcal{H}\right|$ is the length, area or volume of
$\mathcal{H}$ for $l=1,2,\ \mathrm{and}\ 3,$ respectively. Further,
$\mathbb{P}\left(\mathcal{H}\right)=\frac{\lambda_{i}}{\lambda_{0}}$
is independent of $\mathcal{H}$ and hence, $\mathbb{P}\left(\left\{ \left.K=\kappa_{i}\right|\mbox{1 BS in }\mathcal{H}\right\} \right)=\frac{\lambda_{i}}{\lambda_{0}}.$
\end{proof}
The following remarks result due to Theorem \ref{thm:multiTierTheorem}.
\begin{rem}
\label{rem:multiTierEqsingleTier}The equivalent \textit{homogeneous}
$l$-D SCS in Theorem \ref{thm:multiTierTheorem}, with BS density
$\lambda_{0}$ and i.i.d. random transmission powers can further be
reduced to the \textit{homogeneous} $l$-D SCS with BS density $\lambda_{0}E\left[K^{\frac{l}{\varepsilon}}\right]=\lambda_{0}\sum_{i=1}^{M}p_{i}\kappa_{i}^{\frac{l}{\varepsilon}}$
and unity transmission powers at all BSs.\end{rem}
\begin{proof}
\textit{(Outline)} In $\left(\ref{eq:ctoisf}a\right)$ and $\left(\ref{eq:ctoisf}b\right)$,
$\Psi_{S}$ and $\left\{ \Psi_{i}\right\} _{i=1}^{\infty}$ are equal
to unity; $K_{S}$ and $\left\{ K_{i}\right\} _{i=1}^{\infty}$ are
i.i.d. discrete random variables with the probability mass function
(p.m.f.) of $K$ (Theorem \ref{thm:multiTierTheorem} (b)). Now, follow
the same steps in the proof of Theorem \ref{thm:ShadowFadingHomogeneousInvariance}
with $\bar{R}=RK^{-\frac{1}{\varepsilon}}$ to obtain the result.\end{proof}
\begin{rem}
\label{rem:multiTierCtoiIndependence}For the multi-tier network,
Theorem \ref{thm:ldscs} and Remark \ref{rem:ctoiTailProb} together
give the tail probability of $\frac{C}{I}$, and Remark \ref{rem:ctoiIndependentOfLambda}
shows that the $\frac{C}{I}$ is independent of the $\lambda_{0}$,
$\left\{ p_{i}\right\} _{i=1}^{M}$, and $\left\{ \kappa_{i}\right\} _{i=1}^{M}$.
Further, Corollary \ref{cor:ctoiClosedForm} gives the closed form
expression for the tail probability of $\frac{C}{I}$ in $\left[1,\infty\right)$,
and Theorem \ref{thm:fewBS} gives the closed form expression for
$\frac{C}{I}$ under few BS approximation.
\end{rem}
Since the \textit{homogeneous} $l$-D SCS (homogeneous Poisson point
process in $\mathbb{R}^{l}$) has the maximum entropy for a given
mean number of points in any subset of $\mathbb{R}^{l}$, Remark \ref{rem:multiTierCtoiIndependence}
shows that even the most arbitrary placement of BSs in $\mathbb{R}^{l}$
does not degrade the $\frac{C}{I}$ performance, and hence any intelligent
strategy in BS placement in any of the tiers of the multi-tier network
will only improve the $\frac{C}{I}$.
\begin{rem}
\label{rem:multiTierCtoIN}The $\frac{C}{I+N}$ of the multi-tier
network has the same c.d.f. as that of the equivalent \textit{homogeneous}
$l$-D SCS in Corollary \ref{cor:ctoiInvariance} with $N'=N\left(\lambda_{0}\mathrm{E}\left[K^{\frac{l}{\varepsilon}}\right]\right)^{-\frac{\varepsilon}{l}}.$
Further, the tail probability of $\frac{C}{I+N}$ can be computed
using Corollary \ref{cor:INtoCcharfn} and $\left(\ref{eq:nosf-expr}\right)$
(by replacing $\frac{C}{I}$ with $\frac{C}{I+N}).$

Note that $\frac{C}{I}$ and $\frac{C}{I+N}$ studied in this section,
for a multi-tier network with $M=2$ corresponds to a 2-tier network
with macrocell network and femtocell network with all femtocell BSs
operating in the open access mode. Further, an important consequence
of Remark \ref{rem:multiTierCtoIN} is as follows.\end{rem}
\begin{cor}
\label{cor:multiTierPerformanceImprovement}\emph{Inclusion of additional
tiers of wireless networks with low transmission power BSs over an
existing wireless network will only improve the $\frac{C}{I+N}$ performance
of the overall network. Further, as the BS density of the additional
tiers increases, $\frac{C}{I+N}$ performance keeps improving, and
approaches the $\frac{C}{I}$ performance as the BS density approaches
infinity.}\end{cor}
\begin{proof}
The existing wireless network is a single tier network with BS density
$\lambda_{1}$ with a constant transmission power $\kappa_{1}$ at
all the BSs. Hence, the $\frac{C}{I+N}$ has the c.d.f. as the equivalent
homogeneous $l$-D SCS in Corollary \ref{cor:ctoiInvariance} with
$N'=N_{1}=N\lambda_{1}^{-\frac{\varepsilon}{l}}\kappa_{1}^{-1}$.
Now, let $M-1$ additional wireless networks be installed on top of
this single tier network to form an $M$-tier network, with BS densities
$\left\{ \lambda_{i}\right\} _{i=2}^{M-1}$ and constant transmission
powers $\left\{ \kappa_{i}\right\} _{i=2}^{M}$. From Remark \ref{rem:multiTierCtoIN},
this $M$-tier network has the same c.d.f. as the equivalent homogeneous
$l$-D SCS in Corollary \ref{cor:ctoiInvariance} with $N'=N_{2}=N_{1}\left(1+\sum_{i=2}^{M}\frac{\lambda_{i}}{\lambda_{1}}\cdot\left(\frac{\kappa_{i}}{\kappa_{1}}\right)^{\frac{l}{\varepsilon}}\right)^{-\frac{\varepsilon}{l}}\le N_{1}.$
 Thus, the $M$-tier network has a smaller noise power, which leads
to an improved $\frac{C}{I+N}$ performance of the overall network
compared to the existing wireless network (c.f. Fig. \ref{fig:cinr_plot}).
Further, as $\left\{ \lambda_{i}\right\} _{i=2}^{M}$ increases, $N_{2}$
decreases, and converges to zero as at least one of $\left\{ \lambda_{i}\right\} _{i=2}^{M}$
approaches $\infty$. Then, the $\frac{C}{I+N}$ converges, in distribution,
to the $\frac{C}{I}$ of a single tier network (Section \ref{sec:ctoi_singleTier}).
\end{proof}
In a practical cellular system, the macrocell BSs have large transmission
powers in order to provide cellular coverage, and the microcell, picocell
and femtocell BSs have relatively smaller transmission powers. In
this case, Corollary \ref{cor:multiTierPerformanceImprovement} applies;
hence the installation of these networks with low power BSs will not
harm the existing cellular performance and any intelligent strategy
will only improve it.

\begin{rem}
\label{rem:multiTierShadowing}When i.i.d. shadow fading factors $\left\{ \Psi_{i}\right\} $
independent of the BS placement random process are introduced to the
multi-tier network, $\frac{C}{I}$ is the same as in a multi-tier
network without shadow fading and $\frac{C}{I+N}$ has the same c.d.f.
as that of the equivalent \textit{homogeneous} $l$-D SCS in Corollary
\ref{cor:ctoiInvariance} with $N'=N\left(\lambda_{0}\mathrm{E}\left[K^{\frac{l}{\varepsilon}}\right]\times E\left[\Psi^{\frac{l}{\varepsilon}}\right]\right)^{-\frac{\varepsilon}{l}},$
as long as $E\left[\Psi^{\frac{l}{\varepsilon}}\right]<\infty$, where
$E\left[\cdot\right]$ is the expectation operator.\end{rem}
\begin{proof}
\textit{(Outline)} The multi-tier network is equivalent to a single
tier network with BS density $\lambda_{0}$ (Theorem \ref{thm:multiTierTheorem}(a)).
Next, in $\left(\ref{eq:ctoisf}a\right)$ and $\left(\ref{eq:ctoisf}b\right)$,
$K_{S}$ and $\left\{ K_{i}\right\} _{i=1}^{\infty}$ are i.i.d. discrete
random variables with the p.m.f. of $K$ (Theorem \ref{thm:multiTierTheorem}
(b)); $\Psi_{S}$ and $\left\{ \Psi_{i}\right\} _{i=1}^{\infty}$
are i.i.d. random factors with the p.d.f. of $\Psi$. Now, follow
the same steps in the proof of Theorem \ref{thm:ShadowFadingHomogeneousInvariance}
with $\bar{R}=RK^{-\frac{1}{\varepsilon}}\Psi^{-\frac{1}{\varepsilon}}$,
and reduce the equivalent single tier network to another single tier
network with BS density $\lambda_{0}\mathrm{E}\left[K^{\frac{l}{\varepsilon}}\right]E\left[\Psi^{\frac{l}{\varepsilon}}\right]$,
unity transmission powers at all BSs and no shadow fading. Finally,
use Remark \ref{rem:multiTierCtoiIndependence} and Remark \ref{rem:multiTierCtoIN}
to complete the proof.
\end{proof}

\begin{cor}
\emph{(Ideal sectorized antennas) If each BS in the $i^{\mathrm{th}}$
tier of the $M$-tier network has BSs with ideal sectorized antennas
with an antenna gain, $G_{i},$ and beam-width $\theta_{i}$, such
that the BSs antenna faces the MS with probability $\frac{\theta_{i}}{2\pi}$,
in which case the transmission power is $K_{i}=G_{i}\times X_{i}$,
where $X_{i}\sim\mathrm{Bernoulli}\left(\frac{\theta_{i}}{2\pi}\right)$
for each $i=1,2,\cdots,M$, the equivalent homogeneous $l$-D SCS
will have BSs with transmission powers which are i.i.d. random variables
and have a probability mass function (p.m.f.): $\mathbb{P}\left(\left\{ K=G_{i}\right\} \right)=\frac{p_{i}\theta_{i}}{2\pi},\ i=1,2,\cdots,N,\ \mathrm{and}\ \mathbb{P}\left(\left\{ K=0\right\} \right)=1-\sum_{i=1}^{M}\frac{p_{i}\theta_{i}}{2\pi}.$}
\end{cor}
Further, the $\frac{C}{I}$ and $\frac{C}{I+N}$ of this multi-tier
network can be computed using Remark \ref{rem:multiTierCtoiIndependence}
and Remark \ref{rem:multiTierCtoIN}. Finally, Corollary \ref{cor:multiTierPerformanceImprovement}
and Remark \ref{rem:multiTierShadowing} also hold true for this multi-tier
network.

\section{Conclusions}

In this paper, we study the $\frac{C}{I}$ and $\frac{C}{I+N}$ at
the MS within a multi-tier network, where each tier is modeled as
the \textit{homogeneous} $l$-D SCS $\left(l=1,2,\mathrm{\ and\ }3\right)$.
Most studies of wireless networks model the network as the homogeneous
Poisson point process in $\mathbb{R}^{2}$. Here, we study the wireless
network with the BS distribution according the homogeneous Poisson
point process in $\mathbb{R}^{1}$ and $\mathbb{R}^{3}$ as well,
and highlight their significance in practical scenarios.

The $\frac{C}{I}$ and $\frac{C}{I+N}$ in a single tier network are
thoroughly analyzed. Using these results, we completely characterize
the $\frac{C}{I}$ and the $\frac{C}{I+N}$ at the MS within a multi-tier
($M$-tier) network. This paper brings together and refines a set
of results on \textit{homogeneous} $l$-D SCS to demonstrate how the
SCS model (developed in \cite{Brown2000,Madh0000:Carrier,Madhusudhanan2010a})
can easily handle the case of multi-tier networks. The main takeaway
from this paper is due to Corollary \ref{cor:multiTierPerformanceImprovement}:
in a practical cellular system, installation of additional wireless
networks (microcells, picocells and femtocells) with low power BSs
over the already existing macrocell network will always improve the
signal quality at the MS, measured in terms of the tail probability
of $\frac{C}{I+N}$.

\bibliographystyle{IEEEtran}
\bibliography{multiTierNetworks_globecom2011_v3}

\end{document}